\documentclass[12pt,english]{smfart}
\usepackage{amscd,amsfonts,amsmath,amssymb,amstext,amsthm}
\usepackage[]{graphicx}
\title
[Models of disordered bosons]
{A method for constructing random matrix models of
disordered bosons}
\thanks{The first author's research for this article was partially supported by
SFB/TR 12,``Symmetries and Universality in Mesoscopic systems'', of the DFG}
\author{Alan Huckleberry and Kathrin Schaffert}
\address{
Fakult\"at f\"ur Mathematik,\\ 
Ruhr-Universit\"at Bochum,\\ 
Universit\"atsstra\ss e 150,\\ 
D-44801 Bochum, Germany}
\email{ahuck@cplx.ruhr-uni-bochum.de}
\email{kathrin.schaffert@web.de}
\date{\today}
\theoremstyle{plain}
\newtheorem{theorem} {Theorem} [section]
\newtheorem{lemma} [theorem]{Lemma}
\newtheorem{proposition}[theorem]{Proposition}
\newtheorem{corollary} [theorem]{Corollary}
\begin{document}
\begin {abstract}
\noindent
Random matrix models of disordered bosons consist of matrices
in the Lie algebra $\mathfrak g=\mathfrak {sp}_n(\mathbb R)$. Assuming 
dynamical stability, their eigenvalues are required to be
purely imaginary.  Here a method is proposed for constructing
ensembles $(\mathcal E,P)$ of $G$-invariant sets $\mathcal E$
of such matrices with probability measures $P$. These arise as 
moment map direct images from phase spaces $X$ which play an important
role in complex geometry and representation theory. 
In the toy-model case of $n=1$, where $X$ is the
complex bidisk and $P$ is the direct image of the uniform measure,
an explicit description of the spectral measure is given.   
\end {abstract}
\maketitle
The goal of this article is to point out a method for constructing
ensemble probability spaces in a bosonic setting.  We do this in
the context of the model in \cite {LSZ} where due to the requirement
of dynamical stability the
generators 
\begin {gather*}
S=
\begin {pmatrix}
A & -B\\
C & -A^t
\end {pmatrix}
\end {gather*}
of the time evolution are symplectic matrices of elliptic
type, i.e., their eigenvalues come in pairs $\pm i\lambda _j$ with
$\lambda _j>0$.  We denote the set of these elements of the
symplectic Lie algebra $\mathfrak g =\mathfrak {sp}_{2n}(\mathbb R)$ by
$\mathcal E$. 

\bigskip\noindent
To put this in perspective it is perhaps useful to
recall that in the case of fermions the basic model spaces consist
of Hermitian operators $H$ such that $iH$ is, for example, an element 
of one of the Lie algebras $\mathfrak {su}_n$, $\mathfrak {so}_n$ or
$\mathfrak {usp}_{2n}$ of the classical compact groups.  
For example, in the case of $\mathfrak {su}_n$ the basic model space 
becomes the vector space of all Hermitian matrices.  In that case
the model probability distribution can be chosen to be Gaussian,
up to constants the density being $f(H):=e^{-\Vert H\Vert^2}$.  Since
one is usually most interested in associated eigenvalue distributions,
it is entirely appropriate that the model densities are invariant 
with respect to the action by conjugation of the compact group $K$
at hand, i.e., the adjoint representation of $K$.  In the case of 
$\mathfrak {su}_n$ this is just the action of the special unitary 
group $K=\mathrm{SU}_n$. Reducing the symmetry by this action, one 
obtains associated classical
Weyl-group invariant distributions for the eigenvalues. One then
computes various associated distributions, e.g., eigenvalue spacing
distributions and their limits.  For example, in the case of
the Hermitian matrices the famous GUE-density appears as a 
limiting distribution in this way. 

\bigskip\noindent
At the very outset, one observes two major differences between
the setting of disordered bosons and that of fermions.  
First, the necessity of dynamic stability imposes a
strong condition on the operators being considered.  Secondly, unlike the
fermionic case where the model probability densities can be chosen
to be invariant with repect to the full compact group of symmetries, in the
bosonic case the group $G=\mathrm {Sp}_{2n}(\mathbb R)$ is noncompact
and $G$-invariant densities are not available. 

\bigskip\noindent
In \cite {LSZ} an interesting class of bosonic emsembles $(\mathcal E,P)$ is
introduced and the program which we sketched above for fermions is carried
out.  Here, as in \cite {LSZ}, ensemble densities which are invariant
with respect to the (unique up to conjugation) maximal compact
subgroup $K$ of $G$ are proposed.  The usual choice is
$K=U_n$. Our hope here is that $K$-invariant constructions 
which are closely related to the geometry and represention theory 
of the symplectic group $G$ will also produce bosonic ensembles
of quantum mechanical interest.  We begin by explaining this
approach for arbitrary $n$. Calculations are carried out in
what might seem to be a toy model, i.e., the case of $n=1$.
Due to the fact that we are dealing with a
group of Hermitian type, where the Poincar\'e disk plays a
fundamental role in understanding the associated symmetric space,
this is an important special case in our
future work.  

\section {Generalities}
From the mathematical point of view there is a very interesting 
canonically defined neighborhood  $\mathcal {U}(G)$ of the 
0-section of the cotangent bundle $T^*M$ of the associated 
symmetric space $M=G/K$.  In the terminology of 
\cite{FHW}, to which we refer for background, it is called
the \emph{universal domain} associated to $G$.  Its importance
in representation theory, where it is most often called the
\emph{crown of the symmetric space}, has been underlined 
in numerous works.  In the case at hand, where $G$ is of Hermitian
type, this is $G$-equivariantly and holomorphically isomorphic
to the product $B\times \overline {B}$ of the associated
Hermitian bounded domain and its complex conjugate $\overline {B}$
(\cite {BHH}). It should be noted that $B$ and $\overline B$ 
are biholomorphically equivalent as complex manifolds, but
are not $G$-equivariantly biholomorphic.  

\subsection* {Background on Hermitian symmetric spaces}
In order to emphasize the concrete nature of this situation let
us provide a sketch of some basic information about the 
Hermitian symmetric space $B$ (see, e.g., \cite{H} for details).
A symmetric space is simply a (finite-dimensional) Riemannian manifold
$(M,g)$ with the additional property that at every $p\in M$ there
exists an isometry $\sigma_p$ which is a symmetry at $p$ in the sense
that $\sigma _p^2=\mathrm{Id}_M$ and $d\sigma _p=-\mathrm{Id}_{T_pM}$.
One shows that the connected component of the identity of the
isometry group is a Lie group $G$ which acts transitively on $M$
and that the isotropy subgroup at a neutral point under consideration
is a compact subgroup $K$. 

\bigskip\noindent
Based on curvature conditions and
canonical decompositions of Lie groups, symmetric spaces break up
into products of symmetric spaces of three irreducible types.
One such type is $M=G/K$ where $G$ is a simple (real) noncompact Lie group
and $K$ is a (unique up to conjugation) maximal compact subgroup.
In some cases these symmetric spaces $M$ have the structure of a complex
manifold such that the group $G$ acts as a group of holomorphic 
transformations.  This situation can be group-theorectically characterized
by the condition that $K$ has a positive-dimensional center which
in fact turns out to be $S^!$.  In this case the metric can be
chosen to be Hermitian and therefore one refers to $M$ as
a Hermitian symmetric space and $G$ being of Hermitian type. 
Note that the maximal compact subgroup
$K=U_n$ of $G=\mathrm {Sp}_n(\mathbb R)$ satifies this condition.

\bigskip\noindent
Noncompact Hermitian symmetric spaces of the type discussed above 
are naturally realized as distinguished open $G$-orbits in compact
Hermitian symmetric spaces $Z$ which are referred to as the compact
duals of the noncompact Hermitian symmetric spaces. 
A simple example is the case where
$G=\mathrm {SU}(n,1)$ acts on the space $\mathbb P_n(\mathbb C)$ 
of lines in $\mathbb C^{n+1}$. Here an associated symmetric space
is the open $G$-orbit of negative lines.  In general there
are always exactly two $G$-invariant complex structures on a
Hermitian symmetric space.  In this case the ``other'' structure
is space of positive $n$-planes in $\mathbb P_n$.      

\bigskip\noindent
Using a precise description of the action of
the complex Lie group $G^\mathbb C$ on $Z$, one observes that
the noncompact Hermitian symmetric space is contained as
a bounded domain in a canonically determined (dense) open subset
of $Z$ which is biholomorphically equivalent to $\mathbb C^m$ where
$m:=\dim_\mathbb C M$.  Conversely, every (irreducible) bounded
domain $D$ in $\mathbb C^m$ which is symmetric with respect to
holomorphic symmetries arises in this way.  In our case we may view
the Hermitian symmetric spaces $B$ and $\overline {B}$ as open 
$G$-orbits in the space $Z$ of Lagrangian subspaces of 
$(\mathbb C^{2n},\omega )$, where
$\omega =dp\wedge dq$ is the standard (holomorphic) symplectic structure.

\bigskip\noindent
As a complex manifold $X=B\times \overline {B}$ can also be $G$-equivariantly
realized as a $G$-invariant open neighborhood of the (totally real)
orbit $G.z_0$ of the neutral point in the affine symmetric space
$G^{\mathbb C}/K^{\mathbb C}$.  As mentioned above, it is $G$-equivariantly 
identifiable with a neighborhood of the 0-section in $T^*M$.  These
identifications allow us to consider natural $G$-invariant symplectic
structures on $X$ and the resulting moment maps $\mu :X\to \mathfrak g^*$.   
The most obvious of these are the restriction to $X$ of the
canonical symplectic structure on $T^*M$ and structures of the
type $\omega := i\partial \bar \partial \rho $, where $\rho $
is a strictly plurisubharmonic $G$-exhaustion of $X$ defined, e.g., 
by representation theory.

\bigskip\noindent
One can show that with respect to any of these structures the
generic $G$-orbits are coisotropic, i.e., that the $\mu$-fibers
are contained in the $G$-orbits.  Thus an associated moment map
$\mu $ transports a 
setting of geometric and representation importance to its
image $\mu (X)$ in an optimally controllable way. Our goal
is to construct bosonic ensembles by equipping such an image
with $\mu $-direct image measures.  However, for the symplectic
structures mentioned above, $\mu (X)$ is never contained
in $\mathcal E$.

\bigskip\noindent
The fact that this happens, which was initially surprising to us,
should be further investigated.  Here, however, we move to
a slightly different complex geometric situation by implementing the
$G$-equivariant complex conjugation $\kappa :\bar B\to B$ and
thereby redefining $X=B\times B$.  Thus we still have the
advantage of the symplectic geometry and representation theory
of the original setup. In addition, in this new situation the moment map $\mu $
associated to any $G$-invariant K\"ahlerian structure on $X$        
satisfies $\mu (X)\subset \mathcal E$.  Furthermore, $\mu $ is generically
an open map and consequently its image is not pathological. The
discussion of the $\mu $-images, both before and after complex conjugation,
is contained in $\S\ref{Kaehler}$.

\bigskip\noindent
In the following section we compute $\mu $ for a particularly
beautiful symplectic structure on the product $X=\Delta \times \Delta $,
where $\Delta $ is the unit disk in the complex plane. Here we use
the identification of $G$ with $\mathrm{SU}(1,1)$. This is
certainly a toy model.  However, the
bounded domain $B$ contains a (flat) r-dimensional complex polydisk $\Delta_r$
defined by a choice of a maximal (abelian) subalgebra $\mathfrak a$
in the $\mathfrak p$-part of a Cartan decomposition 
$\mathfrak g=\mathfrak k\oplus \mathfrak p$ such that 
$K.\Delta _r=B$. Therefore certain results can be immediately carried
over to $B$ by arguing one variable at a time in $\Delta_r$.  
On the other hand, our goal of understanding
asymptotic information concerning spectral invariants is of course
not touched here.  Nevertheless we do indicate how to compute the
distribution function of the eigenvalue $i\omega$ with $\omega>0$, 
and we compute this explicitly in the concrete case of the 
direct image of the uniform measure.      

\section {The $\mathbf {SU(1,1)}$-toy model}\label{Kaehler}
Define the mixed signature Hermitian form $\langle \, ,\, \rangle_{1,1}$
on $\mathbb C^2$ by $\langle z,w\rangle _{1,1}:=\bar{z}_1w_1-\bar {z}_2w_2$ and
denote by $G=\mathrm {SU}(1,1)$ its group of complex linear isometries. 
As usual $\mathfrak g:=\mathrm {Lie}(G)$. The group has exactly three
orbits on the projective space $\mathbb P_1(\mathbb C)$:
$\Delta :=G.[0:1]$, $\bar {\Delta }=G.[1:0]$ and the real projective space
$\mathbb P_1(\mathbb R)=G.[1:1]$.  Note that the 
involution $\kappa :\mathbb P_1\to \mathbb P_1$, 
$[z_1,z_2]\mapsto [\bar{z}_2:\bar {z}_1]$ defines a $G$-equivariant,
antiholomorphic isomorphism $\kappa :\bar {\Delta}\to \Delta $.

\bigskip\noindent
For moment-map considerations it is convenient to define the
symmetric bilinear form $b:\mathfrak {g}\times \mathfrak {g}\to \mathbb R$
by $b(x,y):=\frac{1}{2}\mathrm {Tr}(xy)$ which is invariant under
the adjoint representation of $G$. Using the duality defined by
$b$ we will regard the moment map as having image in $\mathfrak g$.
We will use the following 
$b$-orthogonal basis of $\mathfrak g$:
\begin {gather*}
\xi =i
\begin {pmatrix}
1 & 0\\
0 & -1
\end {pmatrix}
\, , \ \eta =i
\begin {pmatrix}
0 & -1\\
1 & 0
\end {pmatrix}
\ \text{and} \ \zeta=
\begin {pmatrix}
0 & 1\\
1 & 0
\end {pmatrix}
\,.
\end {gather*} 
If a moment map with values in $\mathfrak g^*$ is given
by $\mu =\mu_\xi\xi^*+\mu_\eta \eta ^*+\mu_\zeta \zeta ^*$,
then the corresponding map with values in $\mathfrak g$ 
is given by $\mu =-\mu_\xi\xi+\mu_\eta \eta +\mu_\zeta \zeta $.
\subsection* {K\"ahlerian structure}
Here we make two observations which make
a case for considering the $B\times B$ instead of $B\times \overline {B}$
as the appropriate phase space. 
\begin {proposition}\label{negative}
The image of a moment map $\mu :\Delta \times \bar {\Delta}\to \mathfrak g^*$ 
which is defined either by the restricted coadjoint structure or by
a K\"ahler form $\omega =i\partial \bar \partial \rho $ of
a $G$-invariant potential is contained in the complement of the
cone $\mathcal E$ of elliptic elements.
\end {proposition}
For a complementary statement for $X=\Delta \times \Delta $ we
denote by $\hat {\mathcal E}$ the augmentation of either the positive
or negative cone of elliptic elements by adding $0\in \mathfrak g$.
\begin {proposition}\label{positive}
The image of a moment map $\mu :\Delta \times {\Delta}\to \mathfrak g^*$ 
which is defined by a K\"ahler form $\omega =i\partial \bar \partial \rho $ of
a $G$-invariant potential is contained in an augmented cone 
$\hat{\mathcal E}$ with the diagonal being mapped to $0$ and
the complement of the diagonal being mapped to the component
of $\mathcal E$ which is contained in $\hat {\mathcal E}$.
\end {proposition}
It should be noted that these observations can be translated to
the higher dimensional setting by using strongly orthogonal
roots and the polydisk slice mentioned above.

\bigskip\noindent
Turning to the proofs of these remarks, we first note that
in the case where the K\"ahler form is defined by an invariant 
potential, an associated moment map is given by
\begin {equation}\label{definition}
\mu (x)(\xi) =:\mu_\xi(x)=-J\hat {\xi}(x)(\rho )\,.
\end {equation}
Here $J$ is the complex structure and $\hat {\xi}$ denotes
the vector field on the manifold which is associated to
$\xi \in \mathfrak g$.  Recall that since $G$ is semisimple,
the moment map is the unique equivariant map to $\mathfrak g^*$ with
the Hamiltonian property $d\mu_\xi=\iota_{\hat \xi}\omega $.  Of
course $\hat \xi$ is tangent to the level sets of the invariant
function $\rho $ and the moment map is measuring how $\rho $ is
growing along $-J\hat \xi $.  In this situation it is therefore
important to determine the complex tangent spaces  
$T^{CR}_x\{\rho =\rho (x)\}$ of the $\rho $-level sets. Assuming
$d\rho (x)\not=0$, this Cauchy-Riemann tangent space is the 
1-codimensional maximal complex subspace of the (real) 
tangent space $T_x\{\rho =\rho (x)\}$.  

\bigskip\noindent
\textit{Proof of Proposition \ref{negative}}. Let 
$z_0\in \Delta\times \overline {B}$ be
the base point so that $M=G.z_0=G/K$ is the unique symmetric
space orbit in $\Delta \times \bar {\Delta}$. In the standard
homogeneous coordinates this is defined by $\langle z,w\rangle_{1,1}=0$.
Choose $\mathfrak {a}\subset \mathfrak {p}$ to be spanned by the
matrix $\eta $ above. Let $\Sigma $ be the image of the half-open interval
$[0,\frac{\pi}{4})$ by the map $t\mapsto \exp (ti\eta)$.  Then
$\Sigma $ is an exact slice for the $G$-action on 
$\Delta \times \bar{\Delta }$ (See \cite{FHW} for this and other background
used here.). This means that every $G$-orbit intersects $\Sigma $
in exactly one point. The orbits $M_p:=G.p$ for $p\in \Sigma \setminus \{z_0\}$
are real hypersurfaces. The complex tangent spaces $T^{CR}_pM_p$ are
calculated in a general setting in \cite{FHW} (see p. 103).  Applying this
to the case at hand, we observe that these spaces 
are spanned by $\hat {\nu}(p)$ and
$\hat {\theta \nu}(p)$ where $\nu $ spans the $+2$ root-space defined
by $\eta $ and $\theta $ is the Cartan involution. In concrete
terms $\nu =\xi +\zeta $ and $\theta \nu =\xi -\zeta $. As we have
seen above, it follows that  $\mu_\nu(p)=\mu_{\theta\nu}(p)=0$.  It
follows that $\mu (\Sigma )$ is contained in the $\eta $-axis where
all isotropy groups are noncompact.  In particular, 
the image of $\mu $ is contained in the complement of $\Sigma $.

It remains to prove the same result for the restriction of the
canonical symplectic structure on $T^*M$.  Recall that in this
case $\mu $ is defined by the Liouville form $\lambda $; in particular,
$\mu_\xi(p)=\lambda (p)(\hat {\xi}(p))$.  However, if 
$\pi:T^*M\to M$ is the canonical projection, then 
$\pi_*(p)(\hat {\xi})=0$ and consequently $\mu_\xi(p)=0$ for all
$p\in \Sigma $. Thus again we observe that $\mathrm{Im}(\mu )$
has empty intersection with $\mathcal E$. \qed

\bigskip\noindent
\textit{Proof of Proposition \ref{positive}}. Here we use the
slice $\Sigma $ which is defined as the image of the half-open
interval $[0,1)$ by the mapping $t\mapsto (t,-t)$ in the affine
coordinates $(z,w)$ corresponding to $([1:z], [1:w])$ in 
$Z:=\mathbb P_1\times \mathbb P_1$.  The base point for these
considerations is $0:=(0,0)$.  Just as above we compute the
complex tangent space $T^{CR}_pM_p$ of an orbit $G.p=M_p$ for
$p\in \Sigma \setminus \{0\}$.  For this observe that in this
case $p$ is in the open $G^{\mathbb C}$-orbit in $Z$. Therefore
the complex tangent space $T_pZ$ can be represented as the 
quotient space $\mathfrak {g}^{\mathbb C}/\mathfrak {g}^{\mathbb C}_p$, where 
$\mathfrak {g}^{\mathbb C}_p$ is the (complex) Lie algebra of
the $G^{\mathbb C}$-isotropy group at $p$.  We note that this
is conjugate to the subgroup of diagonal matrices in 
$G=\mathrm {SL}_2(\mathbb C)$. 

A direct calculation (See \cite{S} p. 74) shows that 
$\mathfrak {g}^{\mathbb C}_p$ consists of those matrices
of the form
\begin {gather}
A=
\begin {pmatrix}
0 & ct^2\\
c & 0
\end {pmatrix}
\ \text{with} \ c\in \mathbb C.
\end {gather}
It is then immediate that $\eta $ and $\zeta $ are 
linearly dependent modulo $\mathfrak {g}^{\mathbb C}_p$
when regarded as elements of $\mathfrak {g}^\mathbb C$. 
Thus the real subspace of $T_pZ$ which is generated by
$\hat {\eta}(p)$ and $\hat {\zeta}(p)$ is a complex line
which by dimension arguments \emph{is} the complex tangent
space $T^{CR}_pM_p$.

Finally, by exactly the same argument as in the proof of the
previous theorem it follows that $\mu_\eta (p)=\mu_\zeta (p)=0$
for all $p\in \Sigma$. Since
$\mu =-\mu_\xi\xi +\mu _\eta\eta +\mu_\zeta\zeta $,
it follows that $\mu \Sigma =-\mu_\xi\xi$. Since $0$ is necessarily
a critical point of $\rho $, it follows that $\mu (0)=0$. Standard
computations of $\mathrm {rank}(\mu_*)$ show that $\mu (p)\not=0$
for $p\in \Sigma \setminus \{0\}$ and the desired result follows.\qed

Actually we proved more than was stated in the above proposition.
For future reference we state this here.

\bigskip\noindent
\textbf{Zusatz}. \textit{The image $\mu (\Sigma)$ is contained
in the $\xi $-axis $\{\eta^*=\zeta^*=0\}$.}    
\section {Poincar\'e moment map}
Here we let $d_P:X\to \mathbb R^{\ge 0}$ be the distance function 
defined by the Poincar\'e metric on $\Delta $. We temporarily
remove the diagonal from $X$ where $d_p$ is no longer smooth,
show that $d_p$ is strictly plurisubharmonic so that 
$\omega =i\partial \bar {\partial }d_p$ is a $G$-invariant
K\"ahler form and compute the restriction to the slice $\Sigma $ of the
resulting moment map. As would be expected
this map can be (continuously) extended to $X$ with
the diagonal being mapped to $0$.

\subsection* {Strict plurisubharmonicity of $\mathbf {d_P}$}
It is convenient to change coordinates on $X$, letting $u=z+w$ and $v=z-w$
where $(z,w)$ are the affine coordinates which are used above. Thus
the complex disk $L:=\{u=0\}$ contains the slice $\Sigma $ as a radius.
In particular $L$ is transversal to the $G$-orbit of any point
$p\in \Sigma \setminus \{0\}$. Note that $L$ is invariant by the
diagonal $S^1$-action defined by the element $\xi \in \mathfrak g$.
It therefore follows that a function $f$ on $L$ is strictly plurisubharmonic
if and only if $f(e^x)$ is strictly convex.  Let us begin with this
step which of course requires an explicit computation with $d_P$.

\bigskip\noindent
If
$$
d_S(z,w):=\big \vert \frac{z-w}{1-\bar{w}z}\big \vert
$$
is the Schwarz distance function, then
$$
d_P=\log \frac{1+d_S}{1-d_S}=\mathrm{tanh}^{-1}(d_S)\,.
$$
Since $\mathrm{tanh}^{-1}$ is strictly convex, in order to prove
the following it is enough to show that $f(e^x)$ is strictly convex
where $f$ is the restriction of $d_S$ to $L$.
\begin {lemma} The restriction of $d_P$ to $L$ is
strongly subharmonic.
\end {lemma}
\begin {proof}
Using the obvious coordinates, 
$$
f(e^x)=\frac{2e^x}{1+e^{2x}}\,.
$$
\end {proof}
Now we turn to the ``other'' direction, namely the coordinate
$u$. For this at each point of $L$ we consider the complex
curve $\gamma $ defined by $\gamma (u)=(t+u,-t+u)$.
\begin {lemma}
Near $u=0$ the pullback 
$\delta (u)=d_P(\gamma (u))-d_P(\gamma (0))$ is nonnegative
with $\delta (0)=0$ and $\delta (u)>0$ otherwise. Furthermore,
\begin {equation}\label{positivity}
\frac{d^2 \delta}{du^2}(0)>0\,.
\end {equation}
\end {lemma}
\begin {proof}
As above it is enough to prove the positivity for $d_P$ replaced
by $d_S$.  This is equivalent to 
$$
\vert 1-(u+t)(\bar {u}-t)\vert^2<(1+t^2)^2
$$
for $\vert u\vert $ sufficiently small and nonzero.  This is
done by elementary manipulations (See \cite {S}, p.66-7).
One can also prove (\ref{positivity}) by making several estimates
(Also see \cite {S}, p. 67), but this can be proved in a more conceptual
way: Since the curve $\gamma $ touches the $d_P$ level set from
outside, it is immediate that  $\frac{d^2 \delta}{du^2}(0)\ge 0$. If
this derivative vanished at $p$, then the Levi-form restricted to
the complex tangent bundle of the orbit $G.p$ would vanish identically.
It would then follow that $G.p$ would be foliated by complex curves.
This is for numerous reasons impossible, e.g., continuously
moving the leaf through $p$ away from $p$ so that it no longer
intersects the curve $\gamma $ yields a contradiction 
to Hurwitz's Theorem.
\end {proof}
\begin {proposition}
The Poincar\'e distance function is plurisubharmonic on $X$
and is strictly plurisubharmonic outside of the diagonal.
\end {proposition}
\begin {proof}
The above two lemmas show that 
$$
\frac{\partial^2 d_P}{\partial^2 u}>0 \ \text{and}
\frac{\partial^2 d_P}{\partial^2 v}>0
$$
at every point of $L\setminus \{0\}$. The fact that the curves
$\gamma $ are tangent to the $d_P$-level sets at every point
of $L\setminus \{0\}$ implies that $\frac{\partial d_P}{\partial u}$
vanishes along $L$ and therefore the mixed derivatives
$\frac{\partial^2 d_P}{\partial \bar{v}\partial u}$ vanish
identically there as well.  In other words the
full complex Hessian is diagonalized along $L\setminus \{0\}$
with positive entries along the diagonal.  Since $d_P$ is only
continuous along the diagonal of $X$, to show that it is plurisubharmonic
there, we must prove that if $h:\Delta \to X$ is holomorphic with
$h(0)=0$, then $d_P\circ h$ is subharmonic.  If the image of
$h$ is contained in the diagonal, then this is just the identically
zero function.  Otherwise we may assume that $0$ is the only point
of the diagonal in the image and the result follows from the
plurisubharmonicity of $d_P$ outside the diagonal and that, since
$d_P(0)=0$, the mean value property is fulfilled at $0$.
\end {proof} 
\subsection* {Restricted moment map}
To keep the notation straight let $\Sigma _X$ denote the
slice $\Sigma $ in $X$. The sign of the moment
map chosen so that its image is contained in the augmented cone
$\hat {\mathcal E}$, where $\mathcal E$ is the positive cone
defined as $G.\Sigma_E$ where 
$\Sigma_E:=\mathbb R^{>0}\xi $.  By the above
\textit{Zusatz} we know that 
$\mu \vert \Sigma _X:\Sigma _X\to \Sigma_E$.  Since $\Sigma_X$
and $\Sigma_E$ are slices for the respective $G$-actions, for
our purposes here it is enough to have an explicit description
of the \emph{restricted moment map} $\mu \vert \Sigma_X$.  This
amounts to a description of the function $\mu_\xi$ along $\Sigma _X$.
\begin {proposition}
If $L$ is parameterized by $z\to (z,-z)$ and the ordered basis
$\{\xi,\eta ,\zeta\}$ is used for $\mathfrak g$, then  
$$
\mu \vert L(z)=(\frac{8\vert z\vert}{1-\vert z\vert^2},0,0)\,.
$$
\end {proposition}
\begin {proof}
In this parameterization the restriction of the Poincar\'e
distance funtion to $L$ is given by
$$
d_P(z)=\log \Big(\frac{1+\vert z\vert}{1-\vert z\vert}\Big)\,.
$$
By definition
$$       
\mu_\xi(z)=-J\hat {\xi}(z)d_P= 
-\frac{d}{ds}\Big\vert_{s=0}d_P(e^{is\xi}z)=
-\frac{d}{ds}\Big\vert_{s=0}
\log\Big(\frac{1+e^{-2s}\vert z\vert}{1-e^{-2s}\vert z \vert }\Big)^2
$$
and the desired result follows by a direct computation of the derivative.
\end {proof}
\begin {corollary}
The moment map $\mu :X\to \hat {\mathcal E}$ is surjective.
The diagonal in $X$ is the $\mu$-preimage of the vertex 
and the restriction of $\mu $ to the complement of the diagonal
is a proper $S^1$-bundle.  Furthermore, $\mu $ has the
coisotropic property that $\mu^{-1}(\mu (x))\subset G.x$ for all 
$x\in X$. 
\end {corollary}
\section {Direct image of the uniform measure}  
Let us conclude this note by explicitly computing the spectral density
for the direct image measure $P$ on $E$ defined by the uniform measure $\nu $
on $X$. It appears that other natural densities, e.g., those
related to Gaussian type distributions will only be numerically
computable.  Closed form descriptions of the behaviour of these
at the vertex of the cone should be possible.

Note here that the matrices on the slice $\Sigma_E$ are of
the form 
\begin {gather*}
S=\omega i
\begin {pmatrix}
1 &  0 \\
0 & -1
\end {pmatrix}
\end {gather*}
with $\omega >0$. The relevant distribution is therefore
$F(x):=P(\omega <x)$. We will sketch the explicit computation
of $F$ here, referring to \cite{S} for details.

The mapping $\mu \vert \Sigma _X:\Sigma_X\to \Sigma_E$ is given
by $x=\mu (t)=\frac{8t}{1-t^2}$. Let $t=t(x)$ be the corresponding
inverse value. Thus $F(x)$ is the Euclidean volume of the region
$\{d_P<t(x)\}=R(x)$.  We compute this via fiber integration using the
projection $\pi :X\to \Delta $ on the first coordinate. For
$a\in \Delta $ let $A(a,x)$ be the Euclidean area of the fiber 
$\pi^{-1}(a)\cap R(x)=:D(a,x)$ so that up to a constant
$$
F(x)=\int _\Delta A(a,x)da\wedge d\bar{a}\,.
$$
Let us begin by computing $A(0,x)$.   
\begin {lemma}
The fiber $R(x)\cap \pi^{-1}(0)$ is a disk of radius 
$u=u(t(x))=\frac{2t}{1+t^2}$ in the $\pi$-fiber $\{0\}\times \Delta$.
\end {lemma}
\begin {proof}
Since the projection $\pi $
is $G$-equivariant, $R(x)\cap \pi^{-1}(0)$
is a region in the $\pi$-fiber which is bounded
by an orbit of the $G$-isotropy group at $0$ in the 
$\pi$-image space.  Since this isotropy group is acting on the
fiber by rotations, this is a disk. Now the region $R(x)\cap L$ 
in $L$ has the point $p_u=(u,-u)$ on its boundary. Thus if we
apply the tranformation $g(z,w)=(T_{u}(z),T_{u}(w))$ with
$$
T_{\zeta}(z):=\frac{z-\zeta}{1-\bar {\zeta}w}\,,
$$
then $p_u$ is mapped to the point $(0,T_{u}(-u))$ which
is on the boundary of the disk in question and the desired result
follows.
\end {proof}
Since we are only dealing with the uniform measure on $X$, it follows
that 
$$
F(x)=\frac{i}{2\pi}\int_{a\in \Delta}A(a,x)da\wedge d\bar{a}\,.
$$
Now $D(a,x)$ is the image of the disk $D(0,x)$ by the transformation
$T_a$.  Since $T_a$ preserves the Poincar\'e metric, this is a 
hyperbolic disk, i.e., a disk with respect to $d_P$.  But hyperbolic
disks are Euclidean disks.  Thus, 
$$
F(x)=\frac{i}{2}\int_{a\in \Delta}r^2(a,x)da\wedge d\bar {a}
$$
where $r=r(a,x)$ is the Euclidean radius of $D(a,x)$ which we 
now compute.

By rotational invariance it is enough to compute $r(a,x)$ for
$a=s\in [0,1)$ on the positive real axis of $\Delta $.  Thus the transformation
$T_s$ stabilizes the interval $I$ of real points of $\Delta $. Since
$I$ is orthogonal to the boundary of $D(0,x)$ and $T_s$ is conformal,
it follows that $I$ is orthogonal to the boundary of $D(s,x)$.
Consequently, the length of $I\cap D(s,x)$ is the diameter of 
$D(s,x)$.  This leads to the following fact.
\begin {lemma}
For $s\in [0,1)$ the radius of $D(s,x)$ is
$$
r(s,x)=\frac{1}{2}(T_s(u)-T_s(-u)=\frac{u(1-s^2)}{1-s^2u^2}\,.
$$
\end {lemma}
Introducing polar coordinates and using rotational symmetry, one
shows that 
$$
F(x)=2\int_0^1\Big(\frac{u(1-s^2)}{1-s^2u^2}\Big)^2sds
$$
for $u=u(t(x))$ as above.  This integral can actually be computed
in closed form:
$$
F(x)=2\frac{1-u^2}{u^2}\log (1-u^2)+2-u^2\,.
$$
Using the simple dilation $\tilde x=\frac{x}{4}$ to clean up the numbers,
one computes that
$$
u^2=\frac{\tilde {x}^2}{1+\tilde {x}^2}\,.
$$
Thus one has a rather simple closed form description of 
the distribution function $F$ and the associated density $f$.
\begin {proposition}
In the variable $\tilde x=\frac{x}{4}$ the eigenvalue distribution
function $F$ and its density $f$ associated to the direct image
of the uniform measure by the Poincar\'e moment map
$\mu :X\to \hat {\mathcal E}$ are given by
$$
F=-\frac{2}{\tilde {x}^2}\log(1+\tilde{x}^2) +\frac{1}{1+\tilde{x}^2}
$$
and
$$
f=\frac{4}{\tilde {x}^3}\log (1+\tilde{x}^2)-
\frac{6\tilde {x}^2+4}{\tilde {x}(1+\tilde{x}^2)^2}\,.
$$
\end {proposition}
Numerical computation of $f$ yields the following picture. \\
\\ \\ \\ \\ \\
\\ \\ \\ \\ \\
\begin{center}
		\begin{figure}[h] 
				\scalebox{0.65}[0.65]
{\includegraphics[viewport = 0 30 20 45]{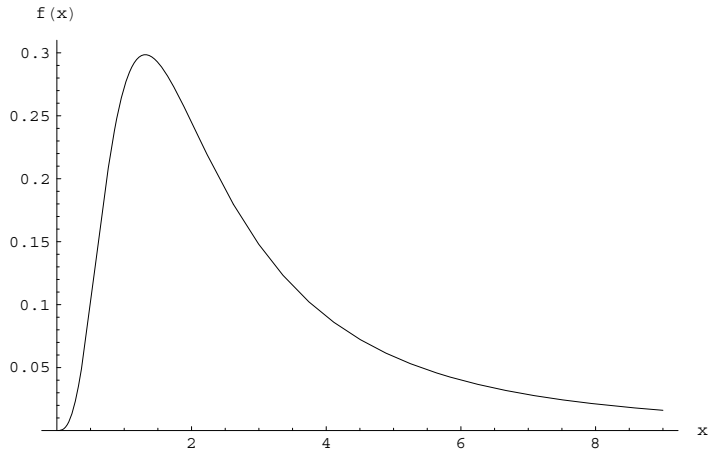}} \\
			\caption{Density function}
		\end{figure}
\end{center}

\bigskip\noindent
Near $0$ one has the power series representation
$$
f(x) = \sum_{k=2}^{\infty}\frac{(-1)^k}{2}
\frac{k(k-1)}{k+1}\left(\frac{x}{4}\right)^{2k-1}\,.
$$
In particular $f\sim x^3$ near $0$.  For $x\sim \infty $
the expression in the above proposition implies that
$f\sim \frac{\log(x)}{x^3}$.  The expected value 
$E(X)=\int_0^\infty xf(x)dx$ can actually be explicitly computed
as $E(X)=\frac{3}{2}\pi $.  However, as one observes from the
estimate of $f$ near $\infty$ the second moment $E(X^2)$ is
infinite. This phenomena does occur in nature, but here it
is probably due to the fact that the original uniform measure
on the phase space $X$ is physically unrealistic.  
Gaussian densities, e.g., of the type $e^{-d_P}$ would perhaps
be of greater interest, but the resulting integrals are more 
complicated.
\end {document}